\documentclass[11pt,a4paper]{amsart}

\usepackage{lineno}
\modulolinenumbers[5]

\usepackage{latexsym,amsfonts,amsmath,amsthm,mathrsfs}
\usepackage{enumitem}

\newtheorem{thm}{Theorem}

\newtheorem{cor}{Corollary}

\theoremstyle{remark}

\theoremstyle{definition}

\newcommand{\disp}{\operatorname{disp}}

\newcommand{\vol}{\operatornamewithlimits{vol}}

\begin{document}

\title{A note on minimal dispersion of point sets in the unit cube} 
\author{Jakub Sosnovec}
\address{Department of Computer Science,
University of Warwick}
\email{j.sosnovec@warwick.ac.uk}
\thanks{This work was supported by the Leverhulme Trust 2014 Philip Leverhulme Prize of Daniel Kr{\'a}l'.}

\begin{abstract}
We study the dispersion of a point set, a notion closely related to the discrepancy. Given a real $r\in(0,1)$ and an integer $d\geq 2$, let $N(r,d)$ denote the minimum number of points inside the $d$-dimensional unit cube $[0,1]^d$ such that they intersect every axis-aligned box inside $[0,1]^d$ of volume greater than $r$. We prove an upper bound on $N(r,d)$, matching a lower bound of Aistleitner et al. up to a multiplicative constant depending only on $r$. This fully determines the rate of growth of $N(r,d)$ if $r\in(0,1)$ is fixed.
\end{abstract}

\maketitle

\section{Introduction}
The geometric discrepancy theory is the study of distributions of finite point sets and their irregularities \cite{matousek}. 
In this note, we study a notion closely related to discrepancy, the dispersion of a point set.

The problem of finding the area of the largest empty axis-parallel rectangle amidst a set of points in the unit square is a classical problem in computational geometry. The algorithmic version has been introduced by Naamad et al. \cite{nlh} and several other algorithms have been proposed over the years, such as \cite{cdl}. The problem naturally generalizes to multi-dimensional variant, where the task is to determine the volume of the largest empty box amidst a set of points in the $d$-dimensional unit cube. 

An active line of research concerns general bounds on the volume of largest empty box for \emph{any} set of points, in terms of the dimension and the number of points. An upper bound thus amounts to exhibiting an example of a point set such that the volume of any empty box is small, while the lower bound asks for the minimal value such that every set of points of given cardinality allows an empty box of that volume. The first results in this direction were given by Rote and Tichy~\cite{rt}. Dumitrescu and Jiang~\cite{DuJi13} first showed a non-trivial lower bound, which was later improved by Aistleitner et al.~\cite{ahr}. An upper bound by Larcher is also given in~\cite{ahr}. Rudolf~\cite{r} has found an upper bound with much better dependence on the dimension. The problem has recently received attention due to similar questions appearing in approximation theory~\cite{nr}, discrepancy theory~\cite{dt,nx} and approximation of $L_p$-norms and Marcinkiewicz-type discretization~\cite{teml1,teml2,teml3}.

The following reformulation is of interest in the applications in approximation theory: If we fix $r\in(0,1)$ to be the ``allowed volume'', how many points in $\mathbb{R}^d$ are needed to force that any empty box has volume at most $r$, in terms of $d$? In other words, we ask for the minimum number of points needed to intersect every box of volume greater than $r$. 
In this note, we establish the optimal asymptotic growth of this quantity for $r$ fixed.

\subsection{Notation}
For a positive integer $k$, let $[k]$ denote the set $\{1,2,\ldots,k\}$. If $x\in\mathbb{R}^d$ is a vector, $(x)_i$ denotes the $i$-th coordinate of $x$. Let $\mathbf{1}$ denote the vector $(1,\ldots,1)\in\mathbb{R}^d$ (where $d$ will always be clear from context). 

For $d\geq 2$, we use $[0,1]^d$ to denote the $d$-dimensional unit cube. A box $B=I_1\times\cdots \times I_d\subseteq [0,1]^d$ is \emph{open} (\emph{closed}) if all of $I_1,\ldots,I_d$ are open (closed) intervals. We define $\mathscr{B}_d$ as the family of all open boxes inside $[0,1]^d$. 

For a set $T$ of $n$ points in $[0,1]^d$, the volume of the largest open axis-parallel box avoiding all points from $T$ is called the \emph{dispersion} of $T$ and is defined as
\begin{equation}\label{def_disp}\disp(T)=\sup_{B\in\mathscr{B}_d,\, B\cap T=\emptyset}\vol(B),\end{equation}
where $\vol(I_1\times\cdots\times I_d)=|I_1|\cdots|I_d|$. Note that the supremum in \eqref{def_disp} is attained, since there are only finitely many inclusion-maximal boxes $B\in\mathscr{B}_d$ avoiding $T$. 

We further define the minimal dispersion for any point set as
\begin{equation}\label{def_disp2}\disp^*(n,d)=\inf_{T\subset[0,1]^d,\,|T|=n}\disp(T).\end{equation}
Again, observe that the infimum in \eqref{def_disp2} is actually attained, since any sequence of $n$-element point sets inside $[0,1]^d$ has a convergent subsequence.

The quantity we mainly consider in this paper is the inverse function of the minimal dispersion,
\[N(r,d)=\min\{n\in\mathbb{N}:\disp^*(n,d)\leq r\},\]
where $r\in(0,1)$. Determining $N(r,d)$ thus corresponds to the question of how many points are needed to intersect every box of volume greater than $r$.

We remark that the functions $\disp^*(n,d)$ and $N(r,d)$ are of course tightly connected and any bounds on them translate between each other. 

\subsection{Previous work}
The trivial lower bound on $\disp^*(n,d)$ is $1/(n+1)$, since we can split the cube into $n+1$ parts and use the pigeonhole principle. This was improved in \cite{ahr} to
\begin{equation}\label{AHRineq}
\disp^*(n,d)\geq \frac{\log_2d}{4(n+\log_2d)}.
\end{equation}

The inequality \eqref{AHRineq} can be reformulated to give a lower bound on $N(r,d)$ for $r\in(0,1/4)$,
\begin{equation}\label{AHRineq2}N(r,d)\geq \frac{1-4r}{4r}\log_2 d.\end{equation}
In order to show \eqref{AHRineq}, the same authors prove an auxiliary lemma, which is equivalent to that
\begin{equation*}
N(1/4,d)\geq \log_2(d+1).
\end{equation*}
Thus for $r\in(0,1/4]$ fixed, we have $N(r,d)=\Omega(\log d)$.

On the other hand, Larcher (with proof also given in \cite{ahr}) has found the following upper bound of the right order in $n$,
\begin{equation*}
\disp^*(n,d)\leq\frac{2^{7d+1}}{n}.
\end{equation*}
For large enough $n$, this is the best bound. However, the dependence on $d$ is exponential, which was recently greatly improved by Rudolf~\cite{r} as follows, 
\begin{equation}\label{rudolf_ineq}
\disp^*(n,d)\leq\frac{4d}{n}\log\left(\frac{9n}{d}\right).
\end{equation}

The best upper bound on the inverse of the minimal dispersion for $r\in(0,1/4]$ fixed, which is the setting of this paper, is the reformulation of~\eqref{rudolf_ineq},
\[N(r,d)\leq 8d q\log(33q),\]
where $q=1/r$. Thus for $r\in(0,1/4]$ fixed, we have $N(r,d)=O(d)$.

\subsection{Our results}
We are interested in determining the asymptotic growth of $N(r,d)$ for $r\in(0,1)$ fixed and $d$ tending to infinity. As it turns out, the rate of growth is different for $r\in(0,1/4]$ and $r\in(1/4,1)$. 
    
First, we show that for $r\in(1/4,1)$, the number $N(r,d)$ is in fact bounded by a constant depending only on $r$. This is in sharp contrast with \eqref{AHRineq2} which implies that $N(r,d)\rightarrow\infty$ for $d\rightarrow\infty$ if $r<1/4$.
\begin{thm}\label{thm1}
For every $r\in(1/4,1)$, there exists a constant $c_r\in\mathbb{N}$ such that for every $d\geq 2$,
\[N(r,d)\leq c_r.\]
In particular, $c_r$ can be set as \[\biggl\lfloor \frac{1}{(r-1/4)}\biggr\rfloor + 1.\]
\end{thm}
Then we show an upper bound for $r\in(0,1/4]$ which matches the lower bound \eqref{AHRineq2} up to a multiplicative constant depending only on $r$.
\begin{thm}\label{thm2}
For every $r\in(0,1/4]$, there exists a constant $c'_r\in\mathbb{R}$ such that for every $d\geq 2$,
\[N(r,d)\leq c'_r\log d.\]
In particular, $c'_r $ can be set as $q^{q^2+2}(4\log q + 1)$, where $q=\lceil 1/r\rceil$.
\end{thm}
This fully determines the asymptotic growth of $N(r,d)$ in terms of $d$.
\begin{cor}
Let $r\in(0,1)$ be fixed and $d$ tend to infinity. If $r\in(0,1/4]$, then $N(r,d)=\Theta(\log d)$, otherwise $N(r,d)=O(1)$.
\end{cor}

\section{Proofs}

\begin{proof}[Proof of Theorem \ref{thm1}]
We proceed with the proof of Theorem \ref{thm1}.

First, consider the case $r\in[1/2,1)$. We set $c_r=1$ and claim that the single central point $\frac{1}{2}\cdot\mathbf{1}\in[0,1]^d$ does not allow an empty box of volume greater than $1/2$. Let $B=I_1\times\cdots\times I_d\in\mathscr{B}_d$ be such that $\frac{1}{2}\cdot\mathbf{1}\notin X$, then there exists a coordinate $i\in[d]$ such that $1/2\not\in I_i$. Hence $|I_i|\leq 1/2$ and the claim follows.

Let us now assume $r\in(1/4,1/2)$. We set
\[\delta=r-1/4>0\text{ and }k_0=\lfloor1/\delta\rfloor\] 
and define the set 
\[X=\{k\delta\cdot\mathbf{1}:k\in[k_0]\}\cup\biggl\{\frac{1}{2}\cdot\mathbf{1}\biggr\}.\]
Note that $X$ is thus a set of points all lying on the diagonal of the unit cube and $|X|\leq k_0 + 1=c_r$. 

Let $B=I_1\times\cdots\times I_d\in\mathscr{B}_d$ be a box such that $B\cap X=\emptyset$. Again, let $i\in[d]$ be such that $1/2\not\in I_i$. If $|I_i|\leq 1/4$, then $\vol(B)\leq 1/4\leq r$, so we can assume $|I_i|>1/4\geq\delta$. Also, we have either $I_i\subset[0,1/2]$ or $I_i\subset[1/2,1]$, without loss of generality assume the former (the argument for the other case is symmetric). 

Let $\alpha\in(0,1/2]$ be the right endpoint of the interval $I_i$ and choose $k\in[k_0]$ maximal so that $k\delta<\alpha$ (such $k$ exists, as $|I_i|>\delta$). Observe that $k\delta\in I_i$, but we assumed $k\delta\cdot\mathbf{1}\not\in B$, therefore there exists an index $j\in[d]$, distinct from $i$, such that $k\delta\notin I_j$. It follows that $|I_j|\leq 1-k\delta$. Finally, by the definition of $k$, we have $\alpha-k\delta\leq\delta$ and hence $|I_j|\leq 1-\alpha+\delta$. We conclude
\[\vol(B)\leq|I_i|\cdot |I_j|\leq\alpha(1-\alpha+\delta)\leq\alpha(1-\alpha)+\delta\leq 1/4+\delta= r.\]
\end{proof}

\begin{proof}[Proof of Theorem \ref{thm2}]
To prove Theorem \ref{thm2}, we use the probabilistic method to construct a set of points that does not allow any empty box of volume greater than $r$. 

Let $q$ be defined as $\lceil 1/r\rceil$. In the following, we will assume that $1/r$ is an integer -- this is without loss of generality, as otherwise we can use the construction for $1/\lceil 1/r\rceil$. 

Let $X$ be the set of 
\[n=q^{q^2+2}(4\log q + 1)\log d\]
points inside $[0,1]^d$ chosen independently and uniformly at random from the grid $\{1/q,2/q,\ldots,(q-1)/q\}^d$. Let $\mathscr{B}_d^r\subset\mathscr{B}_d$ be the set of all boxes of volume greater than $r$. We now have to show that the probability that $X$ intersects every box from $\mathscr{B}_d^r$ is positive.

For a box $B=I_1\times\cdots\times I_d\in\mathscr{B}_d^r$, let the number of coordinates $i\in[d]$ such that $|I_i|\leq 1-r$ be equal to $m_B$. We have
\[r<\vol(B)\leq (1-r)^{m_B},\]
obtaining that $m_B\leq M$ for 
\[M=\biggl\lfloor\frac{\log r}{\log (1-r)}\biggr\rfloor.\]

The critical part of our proof is the following observation. There exists a finite family $\mathscr{Q}_d^r$ of closed boxes inside $[0,1]^d$ such that every open box $B\in\mathscr{B}_d^r$ contains an element of $\mathscr{Q}_d^r$ as a subset. Moreover, the cardinality of $\mathscr{Q}_d^r$ will be at most $(dq)^M$.

To see this, observe that for every box $B=I_1\times\cdots\times I_d\in\mathscr{B}_d^r$, we can find (possibly non-unique) indices $1\leq j_1<j_2<\cdots<j_M\leq d$ and values $k_{j_1},\ldots,k_{j_M}\in[q-1]$ such that 
\begin{enumerate}[label=(\roman*)]
\item $|I_i|> (q-1)/q\text{ for all }i\in [d]\setminus\{j_1,\ldots,j_M\}$ and\label{i}
\item $k_{j_\ell}/q\in I_{j_\ell}$ for all $\ell\in[M]$.\label{ii} 
\end{enumerate}
We can obtain \ref{i} from the fact that $(q-1)/q=1-r$ and the definition of $m_B$, and \ref{ii} from $|I_i|>r=1/q$ for all $i\in[d]$. 

For every such combination of $j_1,\ldots,j_M\in[d]$ and $k_{j_1},\ldots,k_{j_M}\in[q-1]$, we will include in $\mathscr{Q}_d^r$ the closed box $Q=J_1\times\cdots\times J_d\subset [0,1]^d$ with 
\[J_i=
\begin{cases}
[1/q,(q-1)/q]&\text{ if }i\in [d]\setminus\{j_1,\ldots,j_M\},\\
\{k_i/q\}&\text{ otherwise.}
\end{cases}\]

We claim that $\mathscr{Q}_d^r$ satisfies the desired condition. For every $B\in\mathscr{B}_d^r$ and corresponding indices $j_1,\ldots,j_M\in [d]$ and values $k_{j_1},\ldots,k_{j_M}\in[q-1]$, the closed box $Q$ as defined above is contained within $B$ by \ref{i} and \ref{ii}. Moreover, we can bound the cardinality of $\mathscr{Q}_d^r$ as follows, 
\[|\mathscr{Q}_d^r|=\binom{d}{M}(q-1)^M\leq(dq)^M.\] 

By the property of $\mathscr{Q}_d^r$, if $X$ intersects every box in $\mathscr{Q}_d^r$, then it also intersects every box in $\mathscr{B}_d^r$. We thus only need to show that with positive probability, $X$ intersects every box in $\mathscr{Q}_d^r$. 

Let us fix indices $j_1,\ldots,j_M\in[d]$ and values $k_{j_1},\ldots,k_{j_M}\in[q-1]$ and the corresponding closed box $Q=J_1\times\cdots\times J_d\in\mathscr{Q}_d^r$. Clearly, $J_i=[1/q,(q-1)/q]$ implies that any choice of $(x)_i$ will intersect $J_i$. Hence, the only coordinates restricting intersections with $X$ are $j_1,\ldots,j_M$. Moreover, for each of $(x)_{j_1},\ldots,(x)_{j_M}$, there is exactly one choice of $k\in[q-1]$ such that $k=k_{j_\ell}$.

We obtain that for a point $x\in\{1/q,2/q,\ldots,(q-1)/q\}^d$ chosen uniformly at random,
\[\mathbb{P}[x\in Q]= \biggl(\frac{1}{q-1}\biggr)^M>r^M.\]
Therefore, we get
\[\mathbb{P}[X\cap Q = \emptyset]<(1-r^M)^n\leq\exp(-nr^{M}),\]
and, by the union bound,
\begin{equation}\label{eqqq}
\begin{split}
\mathbb{P}[X\cap Q' = \emptyset\text{ for some }Q'\in\mathscr{Q}_d^r] & <(dq)^M\exp(-nr^M).
\end{split}
\end{equation}

If the right hand side of \eqref{eqqq} is bounded by one, then there exists a choice of $X$ such that it intersects every box in $\mathscr{Q}_d^r$. We bound the logarithm of the right hand side. Using the fact that $M \leq q^2$ and the definition of $n$, we obtain 
\begin{equation*}\begin{split}
M\log d+M\log q - nr^M & \leq q^2\log d + q^2\log q - nq^{-q^2}\\
                       & = q^2\log d + q^2\log q - q^2(4\log q + 1)\log d\\
                       & = (1 - 4\log d)q^2\log q \leq 0,
\end{split}
\end{equation*}
where the last inequality follows from $d\geq 2$. This concludes the proof of Theorem \ref{thm2}.
\end{proof}

\section{Acknowledgements}

The author would like to thank Dan Kr{\'a}l' for introducing him to the problem and Jan Vyb{\'i}ral for some initial thoughts and a major simplification of the proof of Theorem \ref{thm2}. The author would also like to acknowledge the discussions held in Oberwolfach Workshop Perspectives in High-Dimensional Probability and Convexity, as shared to him by Jan Vyb{\'i}ral.

  \bibliographystyle{abbrv}
  \bibliography{Dispersion_arxiv}

\begin{thebibliography}{10}

\bibitem{ahr}
C.~Aistleitner, A.~Hinrichs, and D.~Rudolf.
\newblock On the size of the largest empty box amidst a point set.
\newblock {\em Discrete Applied Mathematics}, pages~--, 2017.

\bibitem{cdl}
B.~Chazelle, R.~L. Drysdale, and D.~T. Lee.
\newblock Computing the largest empty rectangle.
\newblock {\em SIAM J. Comput.}, 15(1):300--315, 1986.

\bibitem{dt}
M.~Drmota and R.~F. Tichy.
\newblock {\em Sequences, Discrepancies and Applications}, volume 1654 of {\em
  Lecture Notes in Mathematics}.
\newblock Springer-Verlag Berlin Heidelberg, 1997.

\bibitem{DuJi13}
A.~Dumitrescu and M.~Jiang.
\newblock On the largest empty axis-parallel box amidst n points.
\newblock {\em Algorithmica}, 66(2):225--248, Jun 2013.

\bibitem{matousek}
J.~Matou{\v s}ek.
\newblock {\em Geometric Discrepancy}, volume~18 of {\em Algorithms and
  Combinatorics}.
\newblock Springer-Verlag Berlin Heidelberg, 1999.

\bibitem{nlh}
A.~Naamad, D.~Lee, and W.-L. Hsu.
\newblock On the maximum empty rectangle problem.
\newblock {\em Discrete Applied Mathematics}, 8(3):267--277, 1984.

\bibitem{nx}
H.~Niederreiter and C.~Xing.
\newblock Low-discrepancy sequences and global function fields with many
  rational places.
\newblock {\em Finite Fields and Their Applications}, 2(3):241 -- 273, 1996.

\bibitem{nr}
E.~Novak and D.~Rudolf.
\newblock Tractability of the approximation of high-dimensional rank one
  tensors.
\newblock {\em Constructive Approximation}, 43(1):1--13, Feb 2016.

\bibitem{rt}
G.~Rote and R.~Tichy.
\newblock Quasi-monte-carlo methods and the dispersion of point sequences.
\newblock {\em Mathematical and Computer Modelling}, 23(8):9 -- 23, 1996.

\bibitem{r}
D.~Rudolf.
\newblock An upper bound of the minimal dispersion via delta covers.
\newblock 2017.
\newblock Preprint, arXiv:1701.06430.

\bibitem{teml2}
V.~N. Temlyakov.
\newblock The marcinkiewicz-type discretization theorems.
\newblock 2017.
\newblock Preprint, arXiv:1703.03743.

\bibitem{teml1}
V.~N. Temlyakov.
\newblock The marcinkiewicz-type discretization theorems for the hyperbolic
  cross polynomials.
\newblock 2017.
\newblock Preprint, arXiv:1702.01617.

\bibitem{teml3}
V.~N. Temlyakov.
\newblock Universal discretization.
\newblock 2017.
\newblock Preprint, arXiv:1708.08544.

\end{thebibliography}

\end{document}